\documentclass[11pt]{article}
\usepackage{amssymb}
\usepackage[perpage,symbol]{footmisc}
\usepackage{listings}
\usepackage{amsfonts}
\usepackage{enumerate}
\usepackage{amsmath}
\usepackage{cite}
\usepackage{array}
\usepackage{booktabs}
\newcommand\relphantom[1]{\mathrel{\phantom{#1}}}

\begin{document}

\newtheorem{theorem}{Theorem}[section]
\newtheorem{corollary}[theorem]{Corollary}
\newtheorem{definition}[theorem]{Definition}
\newtheorem{proposition}[theorem]{Proposition}
\newtheorem{lemma}[theorem]{Lemma}
\newtheorem{example}[theorem]{Example}
\newenvironment{proof}{\noindent {\bf Proof.}}{\rule{3mm}{3mm}\par\medskip}
\newcommand{\remark}{\medskip\par\noindent {\bf Remark.~~}}

\title{A Class of Five-weight Cyclic Codes and Their Weight Distribution\footnote{This work is supported by the
NSFC under Grants  11071160 and 11001170.}}
\author{Yan Liu\footnote{Corresponding author, Dept. of Math., SJTU, Shanghai, 200240,  liuyan0916@sjtu.edu.cn.}, Haode Yan\footnote{Dept. of Math., Shanghai Jiaotong Univ., Shanghai, 200240, hdyan@sjtu.edu.cn.}
}
\date{}
\maketitle
\thispagestyle{empty}

\abstract{In this paper, a family of five-weight reducible cyclic codes  is presented. Furthermore, the weight distribution of these cyclic codes is determined, which follows from the determination of value distributions of certain exponential sums. }



\section{\small{INTRODUCTION}}
Recall that an $[n,l,d]$ linear code $\mathcal{C}$ over $\mathbb{F}_{q}$ is a linear subspace of $\mathbb{F}_{q}^{n}$ with dimension $l$ and minimum Hamming distance $d$, where $q$ is a prime power. Let $A_{i}$ denote the number of codewords  in $\mathcal{C}$ with Hamming weight $i$.
The sequence $(A_{0}, A_{1}, A_{2},\ldots, A_{n})$ is called the weight distribution of the code $\mathcal{C}$.
And $\mathcal{C}$ is called cyclic if for any $(c_{0}, c_{1},  \ldots, c_{n-1}) \in \mathcal{C}$, also $(c_{n-1}, c_{0},  \ldots, c_{n-2}) \in \mathcal{C}$.
A linear code $\mathcal{C}$ in $\mathbb{F}_{q}^{n}$ is cyclic if and only if $\mathcal{C}$ is an ideal of the polynomial residue class ring $\mathbb{F}_{q}[x]/(x^{n}-1)$. Since $\mathbb{F}_{q}[x]/(x^{n}-1)$ is a principal ideal ring, every cyclic code corresponds to a principal ideal $(g(x))$ of  the multiples of a polynomial $g(x)$ which is the monic polynomial of lowest degree in the ideal. This polynomial  $g(x)$ is called the generator polynomial, and $h(x)=(x^{n}-1)/g(x)$ is called the parity-check polynomial of the code $\mathcal{C}$. We also  recall that a cyclic code over $\mathbb{F}_{q}$  is called irreducible  if its parity-check polynomial  is irreducible over $\mathbb{F}_{q}$ and reducible, otherwise. Determining the weight distribution of a cyclic code is an important research object in coding theory. Information on the weight distribution of binary cyclic codes can be found in \cite{18,28,29,30}. For information on the weight distribution of non-binary cyclic codes, the reader is referred to \cite{1,2,3,4,6,7,9,12,13,14,15,16,20,21,22}. In this paper, we will determine the weight distribution of a class of five-weight reducible non-binary cyclic codes.

Throughout this paper, let $m$ and $k$ be any two positive integers such that $s=m/d \geq 5$ is odd, where $d=$ gcd$(m,k)$. Let $p$ be  an odd prime, $q=p^{m}$ and $q_{0}=p^{d}$. Then we have $q=q_{0}^{s}$. Let $t$ be a divisor of $d$ such that $d/t$ is odd, and $m_{0}=m/t$.

Let $\pi$ be a primitive element of the finite field $\mathbb{F}_{q}$.
Let $h_{0}(x)$, $h_{1}(x)$ and $h_{2}(x)$ be the minimal polynomials of $\pi^{-2}$, $\pi^{-(p^{k}+1)}$ and $\pi^{-(p^{2k}+1)}$ over $\mathbb{F}_{p^{t}}$, respectively.
 The cyclic codes over  $\mathbb{F}_{p^{t}}$ with  parity-check polynomial $h_{0}(x)h_{1}(x)$ has been extensively studied by \cite{12}. Let $\mathcal{C}$ be the cyclic code with parity-check polynomial $h_{0}(x)h_{1}(x)h_{2}(x)$.  In the case of $t=1$, the weight distribution of $\mathcal{C}$  can be obtained  by the results of  \cite{22,20}. The objective of this paper is to  consider the problem for any positive $t$ such that $d/t$ is odd. For $t\geq2$ and even $k$, the case of $2\nmid \frac{m}{(m,k)}$ but $2\mid \frac{m}{(m,k/2)}$ exists, which was not considered neither in \cite{22} nor in \cite{20}. Moreover, the weight distribution of this class of cyclic codes in this special case can not be obtained with the same arguments as  in \cite{22} and \cite{20} in a straightforward generalization.  In this paper, we always assume $d/t$ is odd and will show that $\mathcal{C}$ has five nonzero weights and determine the weight distribution of this class of cyclic codes.

The rest of this paper is  organized as follows. Some preliminaries   will be introduced in Section 2.  A family of cyclic codes and their weight distributions will be given in Section 3.

\section{\small{PRELIMINARIES}}
We follow the notation in Section 1. In this section, we first give a brief introduction to the theory of quadratic forms over finite fields.

For any quadratic form  $F$ over $\mathbb{F}_{q_{0}}$, there exists a symmetric matrix $A$ of order $s$  such that $F(X)= XAX^{T}$, where $X=(x_{1}, x_{2},\ldots,x_{s})\in \mathbb{F}_{q_{0}}^{s}$ and $X^{T}$ denotes the transpose of $X$. Then there exists a non-singular matrix $M$ of order $s$ over $\mathbb{F}_{q_{0}}$ such that $MAM^{T}$ is
a diagonal matrix (see \cite{11}). Under the non-singular linear substitution $X=ZM$ with $Z=(z_{1}, z_{2},\ldots,z_{s})\in \mathbb{F}_{q_{0}}^{s}$, then
$
F(X)=ZMAM^{T}Z^{T}=\sum_{i=1}^{r}d_{i}z_{i}^{2},$
where $r$ is the rank of $F(X)$ and $d_{i} \in \mathbb{F}_{q_{0}}^{\ast}$. Let $\triangle= d_{1}d_{2}\cdots d_{r}$ (we assume $\triangle=0$ when $r=0$). Let $\eta_{0}$ be the quadratic multiplicative character of $\mathbb{F}_{q_{0}}$. Then $\eta_{0}(\triangle)$ is an invariant of $A$ under the action of $M \in GL_{s}(\mathbb{F}_{q_{0}})$.
\begin{lemma}[\cite{22}]\label{Le:2.2}
Let $F(X)$ be a quadratic form in $s$ variables of rank $r$ over $\mathbb{F}_{q_{0}}$, then
\begin{equation*}
\sum_{X \in \mathbb{F}_{q_{0}}^{s}}\zeta_{p}^{Tr^{q_{0}}_{p}(F(X))}=
\begin{cases}
\eta_{0}(\triangle)(-1)^{(d-1)r}q_{0}^{s-\frac{r}{2}},& p\equiv 1\mod 4,\\
\eta_{0}(\triangle)(\sqrt{-1})^{dr}(-1)^{(d-1)r}q_{0}^{s-\frac{r}{2}}, & p\equiv 3\mod 4.
\end{cases}
\end{equation*}
where $\zeta_{p}$ is a primitive $p$-th root of unity.
\end{lemma}
For any fixed $(a,b,c) \in \mathbb{F}_{q}^{3}$, let $Q_{a,b,c}(x)=Tr^{q}_{q_{0}}(a x^{2}+bx^{p^{k}+1}+cx^{p^{2k}+1} )$,  we have the following result.
\begin{lemma}\label{Le:2.3}
For any $(a,b,c) \in \mathbb{F}_{q}^{3}\backslash \{(0,0,0)\}$, $Q_{a,b,c}(x)$ is a quadratic form over $\mathbb{F}_{q_{0}}$ with rank at least $s-4$.
\end{lemma}
\begin{proof}
The proof is similar to the proof of Lemma 2 in \cite{12}, so we omit the details.
\end{proof}


\section{\small{A CLASS OF FIVE-WEIGHT CYCLIC CODES AND THEIR WEIGHT DISTRIBUTION}}
We follow the notation and conditions fixed in Section 1 and 2. It is easy to check that $h_{0}(x)$, $h_{1}(x)$ and $h_{2}(x)$ are polynomials of degree $m_{0}$ and are pairwise distinct when $s\geq 5$. Let $\mathcal{C}$ be the cyclic code with parity-check polynomial  $h_{0}(x)h_{1}(x)h_{2}(x)$. Then $\mathcal{C}$ has length $q-1$ and dimension $3m_{0}$. Moreover, it can be expressed as \[\mathcal{C}=\{\mathbf{c}_{(a,b,c)}: a, b,c \in \mathbb{F}_{q}\},\] where $\mathbf{c}_{(a,b,c)}=\big(Tr^{q}_{p^{t}}(a\pi^{2i}+b\pi^{(p^{k}+1)i}+c\pi^{(p^{2k}+1)i})\big)_{i=0}^{q-2}$. The weight of the codeword $\mathbf{c}_{(a,b,c)}=(c_{0}, c_{1},\ldots,c_{q-2})$  can be expressed by exponential sums in the following way.


\begin{equation}\label{Eq:3.1}
 \begin{split}
 W(\mathbf{c}_{(a,b,c)})
   &= \#\{0\leq i\leq p^{m}-2: c_{i}\neq 0\}\\
   &= q-1-\frac{1}{p^{t}}\sum_{i=0}^{q-2}\sum_{y \in \mathbb{F}_{p^{t}}}\zeta_{p}^{Tr^{p^{t}}_{p}(yc_{i})} \\
   &=  q-1-\frac{1}{p^{t}}\sum_{i=0}^{q-2}\sum_{y \in \mathbb{F}_{p^{t}}}\zeta_{p}^{Tr^{p^{t}}_{p}(y\cdot Tr^{q}_{p^{t}}(a\pi^{2i}+b\pi^{(p^{k}+1)i}+c\pi^{(p^{2k}+1)i}))}\\
   &=   q-1-\frac{1}{p^{t}}\sum_{x \in \mathbb{F}_{q}^{\ast}}\sum_{y \in \mathbb{F}_{p^{t}}}\zeta_{p}^{Tr^{p^{t}}_{p}(y\cdot Tr^{q}_{p^{t}}(ax^{2}+bx^{p^{k}+1}+cx^{p^{2k}+1}))}\\
   &=   p^{m-t}(p^{t}-1)- \frac{1}{p^{t}}\sum_{y \in \mathbb{F}_{p^{t}}^{\ast}}\sum_{x \in \mathbb{F}_{q}}\zeta_{p}^{Tr^{p^{t}}_{p}(y\cdot Tr^{q}_{p^{t}}(ax^{2}+bx^{p^{k}+1}+cx^{p^{2k}+1}))}\\
& =  p^{m-t}(p^{t}-1)- \frac{1}{p^{t}}(\sum_{y \in SQ}\sum_{x \in \mathbb{F}_{q}}\zeta_{p}^{Tr(a(xy^{\frac{1}{2}})^{2}+b(xy^{\frac{1}{2}})^{p^{k}+1}+c(xy^{\frac{1}{2}})^{p^{2k}+1})}\\
&\relphantom{=} {}+\sum_{y \in \overline{SQ}}\sum_{x \in \mathbb{F}_{q}}\zeta_{p}^{Tr^{q_{0}}_{p}(y\cdot Tr^{q}_{q_{0}}(ax^{2}+bx^{p^{k}+1}+cx^{p^{2k}+1}))})\\
& =   p^{m-t}(p^{t}-1)- \frac{1}{p^{t}}(\sum_{y \in SQ}\sum_{x \in \mathbb{F}_{q}}\zeta_{p}^{Tr(a(xy^{\frac{1}{2}})^{2}+b(xy^{\frac{1}{2}})^{p^{k}+1}+c(xy^{\frac{1}{2}})^{p^{2k}+1})}\\
&\relphantom{=} {}+\sum_{y \in SQ}\sum_{x \in \mathbb{F}_{q}}\zeta_{p}^{Tr^{q_{0}}_{p}(\lambda \cdot Tr^{q}_{q_{0}}(a(xy^{\frac{1}{2}})^{2}+b(xy^{\frac{1}{2}})^{p^{k}+1}+c(xy^{\frac{1}{2}})^{p^{2k}+1}))})\\
& = p^{m-t}(p^{t}-1)- \frac{p^{t}-1}{2p^{t}}(\sum_{x \in \mathbb{F}_{q}}\zeta_{p}^{Tr^{q_{0}}_{p}(Q_{a,b,c}(x))}+\sum_{x \in \mathbb{F}_{q}}\zeta_{p}^{Tr^{q_{0}}_{p}(\lambda Q_{a,b,c}(x))}),
\end{split}
\end{equation}
where $SQ$($\overline{SQ}$, resp.) denotes the set of nonzero square elements (non-square elements, resp.) of $\mathbb{F}_{p^{t}}$ and $\lambda$ is a non-square in $\mathbb{F}_{p^{t}}$.
If we define\begin{equation}
 \label{Eq:3.3}
 T(a,b,c)=\sum_{x \in \mathbb{F}_{q}}\zeta_{p}^{Tr^{q_{0}}_{p}(Q_{a,b,c}(x))}+\sum_{x \in \mathbb{F}_{q}}\zeta_{p}^{Tr^{q_{0}}_{p}(\lambda Q_{a,b,c}(x))},
  \end{equation}
then the weight distribution of the code $\mathcal{C}$ is completely determined by the value distribution of $T(a,b,c)$.
Firstly, we have the following lemma.
\begin{lemma}\label{Le:3.2} For any fixed  $(a,b,c) \in \mathbb{F}_{q}^{3}\backslash \{(0,0,0)\}$,
let  $T(a,b,c)$ be defined by (\ref{Eq:3.3}) and $r$ be the rank of $Q_{a,b,c}(x)$.
\begin{itemize}
  \item If $r$ is even, then $T(a,b,c)=\pm 2q_{0}^{s-\frac{r}{2}}$.
  \item If $r$ is odd, then $T(a,b,c)=0$.
\end{itemize}
\end{lemma}
\begin{proof}
According to Lemma \ref{Le:2.2}, we have
\[
\sum_{x \in \mathbb{F}_{q}}\zeta_{p}^{Tr^{q_{0}}_{p}(\lambda Q_{a,b,c}(x))}=\sum_{x \in \mathbb{F}_{q}}\zeta_{p}^{Tr^{q_{0}}_{p}(Q_{a,b,c}(x))}\eta_{0}(\lambda^{r}).
\]
Then $T(a,b,c)=(1+\eta_{0}(\lambda^{r}))\sum_{x \in \mathbb{F}_{q}}\zeta_{p}^{Tr^{q_{0}}_{p}(Q_{a,b,c}(x))}$. Since $d/t$ is odd, $\lambda$ is also a non-square in $\mathbb{F}_{q_{0}}$. Thus, if $r$ is even,
$T(a,b,c)=\pm 2q_{0}^{s-\frac{r}{2}}$ and $0$, otherwise.
This completes the proof.
\end{proof}
\begin{theorem}\label{Th:3.7}
Let $T(a,b,c)$ be defined by (\ref{Eq:3.3}). Then as $(a,b,c)$ runs through $\mathbb{F}_{q}^{3}$, the value distribution of $T(a,b,c)$ is given by Table 1.
\end{theorem}
\begin{table}[htbp]
\caption{Value Distribution of $T(a,b,c)$}
\centering
\begin{tabular}{ll}
 \hline
 Value& Frequency\\
 \hline
 $2p^{m}$ &1\\
 $0$ & $(p^{m}-1)(p^{2m}-p^{2m-d}+p^{2m-4d}+p^{m}-p^{m-d}-p^{m-3d}+1)$\\
 $2p^{\frac{m+d}{2}}$ &  $\frac{(p^{m+d}+p^{(m+3d)/2})(p^{2m}-p^{2m-2d}-p^{2m-3d}+p^{m-2d}+p^{m-3d}-1)}{2(p^{2d}-1)}$\\
 $-2p^{\frac{m+d}{2}}$ &   $\frac{(p^{m+d}-p^{(m+3d)/2})(p^{2m}-p^{2m-2d}-p^{2m-3d}+p^{m-2d}+p^{m-3d}-1)}{2(p^{2d}-1)}$\\
 $2p^{\frac{m+3d}{2}}$ &  $\frac{(p^{m-3d}+p^{(m-3d)/2})(p^{m}-1)(p^{m-d}-1)}{2(p^{2d}-1)}$\\
 $-2p^{\frac{m+3d}{2}}$ & $\frac{(p^{m-3d}-p^{(m-3d)/2})(p^{m}-1)(p^{m-d}-1)}{2(p^{2d}-1)}$\\
 \hline
\end{tabular}
\end{table}
We prove this theorem only for the case that $q_{0}\equiv 3 \pmod 4$. The proof for the case that $q_{0}\equiv 1 \pmod 4$ is similar and omitted. Hence we assume that $q_{0}\equiv 3 \pmod 4$ from now on. In order to determine the value distribution of $T(a,b,c)$, we need a series of  lemmas. Before introducing them, for any positive integer $k$, we define $d_{1}=p^{k}+1$ and $d_{2}=p^{2k}+1$. Since $-1$ is a non-square in $\mathbb{F}_{p^{t}}$ when $q_{0}\equiv 3 \pmod 4$, so in the following, we set $\lambda=-1$.
\begin{lemma}\label{Le:3.9}
Let $q_{0} \equiv 3 \pmod 4$ and let $N_{2}$ denote the number of solutions $(x_{1},x_{2}) \in \mathbb{F}_{p^{m}}^{2}$ of the following system of equations
\begin{equation*}
\begin{cases}
x_{1}^{2}+ x_{2}^{2}=0\\
x_{1}^{d_{1}}+x_{2}^{d_{1}}=0\\
x_{1}^{d_{2}}+x_{2}^{d_{2}}=0.
\end{cases}
\end{equation*}
Then $N_{2}=1$.
\end{lemma}
\begin{proof}
This system of equations have only one solution $(0,0)$, since  $-1$ is a  non-square  in $\mathbb{F}_{p^{m}}$ when $q_{0} \equiv 3 \pmod 4$.
\end{proof}
\begin{lemma}\label{Le:3.8}
Let $q_{0} \equiv 3 \pmod 4$ and let $\overline{N_{2}}$ denote the number of solutions $(x_{1},x_{2}) \in \mathbb{F}_{p^{m}}^{2}$ of the following system of equations
\begin{equation}\label{Eq:2.3}
\begin{cases}
x_{1}^{2}- x_{2}^{2}=0\\
x_{1}^{d_{1}}-x_{2}^{d_{1}}=0\\
x_{1}^{d_{2}}-x_{2}^{d_{2}}=0.
\end{cases}
\end{equation}
Then $\overline{N_{2}}=2p^{m}-1$.
\end{lemma}
\begin{proof}
We observe that $(x_{1},x_{2})$ is a solution of (\ref{Eq:2.3}) if and only if $(x_{1},x_{2})$ is a solution of the first equation of it. So the conclusion follows from the Lemma 6.24 in \cite{11}.
\end{proof}
\begin{lemma}\label{Le:3.10}
Let $q_{0} \equiv 3 \pmod 4$ and let $N_{3}$ denote the number of solutions $(x_{1},x_{2},x_{3}) \in \mathbb{F}_{p^{m}}^{3}$ of the following system of equations
\begin{equation}\label{Eq:2.5}
\begin{cases}
x_{1}^{2}+x_{2}^{2}+ x_{3}^{2}=0\\
x_{1}^{d_{1}}+x_{2}^{d_{1}}+ x_{3}^{d_{1}}=0\\
x_{1}^{d_{2}}+x_{2}^{d_{2}}+ x_{3}^{d_{2}}=0.
\end{cases}
\end{equation}
Then $N_{3}=p^{m+d}+p^{m}-p^{d}$.
\end{lemma}
\begin{proof}

\noindent Case I, when $x_{3}=0$. In this case, by Lemma \ref{Le:3.9}, the number of solutions of (\ref{Eq:2.5}) is  $1$.

\noindent Case II, when $x_{3}\neq 0$. In this case, for any fixed $x_{3}$, the equation system (\ref{Eq:2.5}) has the same number of solutions $(x_{1},x_{2}) \in \mathbb{F}_{p^{m}}^{2}$ as
the system
\begin{equation}\label{Eq:2.6}
\begin{cases}
x_{1}^{2}+x_{2}^{2}+1=0\\
x_{1}^{d_{1}}+x_{2}^{d_{1}}+1=0\\
x_{1}^{d_{2}}+x_{2}^{d_{2}}+1=0.
\end{cases}
\end{equation}
As assumed in the beginning of this paper, we have gcd$(m,k)=d$. Then by the same method as in the proof of Lemma 4.3 in \cite{20}, we can prove that if $(x_{1},x_{2})\in \mathbb{F}_{p^{m}}^{2}$ is a solution of (\ref{Eq:2.6}), then $(x_{1},x_{2}) \in \mathbb{F}_{p^d}^{2}$.
Furthermore, if $(x_{1},x_{2}) \in \mathbb{F}_{p^{d}}^{2}$ is a solution of the first equation of (\ref{Eq:2.6}), then it is a solution of (\ref{Eq:2.6}).
So the number of solutions $(x_{1},x_{2}) \in \mathbb{F}_{p^{m}}^{2}$ of (\ref{Eq:2.6}) is equal to the number of solutions $(x_{1},x_{2}) \in \mathbb{F}_{p^d}^{2}$ satisfying  the first equation of it, which is $p^d+1$ by Lemma 6.24 in \cite{11}. Thus (\ref{Eq:2.6}) has exactly $p^d+1$ solutions.

Summarizing the results of the two cases above, we have that $N_{3}=1+(p^{m}-1)(p^d+1)=p^{m+d}+p^{m}-p^{d}$. This completes the proof.
\end{proof}
\begin{lemma}\label{3.11}
Let $q_{0} \equiv 3 \pmod 4$ and let $\overline{N_{3}}$ denote the number of solutions $(x_{1},x_{2},x_{3}) \in \mathbb{F}_{p^{m}}^{3}$ of the following system of equations
\begin{equation*}
\begin{cases}
x_{1}^{2}+x_{2}^{2}- x_{3}^{2}=0\\
x_{1}^{d_{1}}+x_{2}^{d_{1}}-x_{3}^{d_{1}}=0\\
x_{1}^{d_{2}}+x_{2}^{d_{2}}-x_{3}^{d_{2}}=0.
\end{cases}
\end{equation*}
Then $\overline{N_{3}}=p^{m+d}+p^{m}-p^{d}$.
\end{lemma}
\begin{proof}
The proof is similar to the proof of the lemma above, so we omit the details.
\end{proof}
\begin{lemma}\label{Le:3.12}
Let $q_{0} \equiv 3 \pmod 4$ and let $N_{4}$ denote the number of solutions $(x_{1},x_{2},x_{3},x_{4}) \in \mathbb{F}_{p^{m}}^{4}$ of the following system of equations
\begin{equation}
\begin{cases}
x_{1}^{2}+x_{2}^{2}+ x_{3}^{2}+ x_{4}^{2}=0\\
x_{1}^{d_{1}}+x_{2}^{d_{1}}+ x_{3}^{d_{1}}+ x_{4}^{d_{1}}=0\\
x_{1}^{d_{2}}+x_{2}^{d_{2}}+ x_{3}^{d_{2}}+ x_{4}^{d_{2}}=0.
\end{cases}
\end{equation}
Then $N_{4}=1+(p^{m}-1)(p^{d}+1)(2p^{m}-p^{d}+1)$.
\end{lemma}
\begin{proof}
See Appendix.
\end{proof}
\begin{lemma}\label{Le:3.13}
Let $q_{0} \equiv 3 \pmod 4$ and let $\overline{N_{4}}$ denote the number of solutions $(x_{1},x_{2},x_{3},x_{4}) \in \mathbb{F}_{p^{m}}^{4}$ of the following system of equations
\begin{equation}
\begin{cases}
x_{1}^{2}+x_{2}^{2}+ x_{3}^{2}- x_{4}^{2}=0\\
x_{1}^{d_{1}}+x_{2}^{d_{1}}+ x_{3}^{d_{1}}-x_{4}^{d_{1}}=0\\
x_{1}^{d_{2}}+x_{2}^{d_{2}}+ x_{3}^{d_{2}}- x_{4}^{d_{2}}=0.
\end{cases}
\end{equation}
Then $\overline{N_{4}}=p^{m+2d}+p^{m}-p^{2d}$.
\end{lemma}
\begin{proof}
See Appendix.
\end{proof}
\begin{lemma}\label{Le:3.14}
Let $q_{0} \equiv 3 \pmod 4$ and let $\widetilde{N_{4}}$ denote the number of solutions $(x_{1},x_{2},x_{3},x_{4}) \in \mathbb{F}_{p^{m}}^{4}$ of the following system of equations
\begin{equation}
\begin{cases}
x_{1}^{2}+x_{2}^{2}- x_{3}^{2}- x_{4}^{2}=0\\
x_{1}^{d_{1}}+x_{2}^{d_{1}}- x_{3}^{d_{1}}- x_{4}^{d_{1}}=0\\
x_{1}^{d_{2}}+x_{2}^{d_{2}}- x_{3}^{d_{2}}- x_{4}^{d_{2}}=0.
\end{cases}
\end{equation}
Then $\widetilde{N_{4}}=1+(p^{m}-1)(p^{d}+1)(2p^{m}-p^{d}+1)$.
\end{lemma}
\begin{proof}
See Appendix.
\end{proof}
Now we are ready to prove Theorem \ref{Th:3.7} in the case of $q_{0} \equiv 3 \pmod 4$.

\noindent $\mathbf{Proof}$ $\mathbf{of}$ $\mathbf{Theorem}$ \ref{Th:3.7}.

\noindent It is clear that $T(a,b,c)=2p^{m}$ if $(a,b,c)=(0,0,0)$. Otherwise, by Lemma \ref{Le:2.3} and \ref{Le:3.2}, we have
\[T(a,b,c) \in \{0, \pm 2p^{\frac{m+d}{2}}, \pm 2p^{\frac{m+3d}{2}}\}.\]
We define
\[n_{1,i}=\#\{(a,b,c) \in \mathbb{F}_{p^{m}}^{3}: T(a,b,c)=(-1)^{i}2p^{\frac{m+d}{2}}\},\]
\[n_{2,i}=\#\{(a,b,c) \in \mathbb{F}_{p^{m}}^{3}: T(a,b,c)=(-1)^{i}2p^{\frac{m+3d}{2}}\},\]
where $i=0,1$. Then we immediately have
\begin{equation}\label{Eq:3.10}
\begin{cases}
\sum_{(a,b,c) \in \mathbb{F}_{p^{m}}^{3}}T(a,b,c)=2p^{m}+2(n_{1,0}-n_{1,1})p^{\frac{m+d}{2}}+2(n_{2,0}-n_{2,1})p^{\frac{m+3d}{2}}\\
\sum_{(a,b,c) \in \mathbb{F}_{p^{m}}^{3}}T^{2}(a,b,c)=2^{2}p^{2m}+2^{2}(n_{1,0}+n_{1,1})p^{m+d}+2^{2}(n_{2,0}+n_{2,1})p^{m+3d}\\
\sum_{(a,b,c) \in \mathbb{F}_{p^{m}}^{3}}T^{3}(a,b,c)=2^{3}p^{3m}+2^{3}(n_{1,0}-n_{1,1})p^{\frac{3m+3d}{2}}+2^{3}(n_{2,0}-n_{2,1})p^{\frac{3m+9d}{2}}\\
\sum_{(a,b,c) \in \mathbb{F}_{p^{m}}^{3}}T^{4}(a,b,c)=2^{4}p^{4m}+2^{4}(n_{1,0}+n_{1,1})p^{2m+2d}+2^{4}(n_{2,0}+n_{2,1})p^{2m+6d}.
\end{cases}
\end{equation}
On the other hand, it follows from Eq. (\ref{Eq:3.3}) that
\begin{equation}\label{Eq:3.13}
\begin{split}
&\sum_{(a,b,c) \in \mathbb{F}_{p^{m}}^{3}}T(a,b,c)\\
&=\sum_{(a,b,c) \in \mathbb{F}_{p^{m}}^{3}}(\sum_{x \in \mathbb{F}_{p^{m}}}\zeta_{p}^{Tr(ax^{2}+bx^{d_{1}}+cx^{d_{2}})}+\sum_{x \in \mathbb{F}_{p^{m}}}\zeta_{p}^{Tr(-ax^{2}-bx^{d_{1}}-cx^{d_{2}})})\\
&=\sum_{x \in \mathbb{F}_{p^{m}}}\sum_{a \in \mathbb{F}_{p^{m}}}\zeta_{p}^{Tr(ax^{2})}\sum_{b \in \mathbb{F}_{p^{m}}}\zeta_{p}^{Tr(bx^{d_{1}})}\sum_{c \in \mathbb{F}_{p^{m}}}\zeta_{p}^{Tr(cx^{d_{2}})}\\
&\relphantom{=} {}+\sum_{x \in \mathbb{F}_{p^{m}}}\sum_{a \in \mathbb{F}_{p^{m}}}\zeta_{p}^{Tr(-ax^{2})}\sum_{b \in \mathbb{F}_{p^{m}}}\zeta_{p}^{Tr(-bx^{d_{1}})}\sum_{c \in \mathbb{F}_{p^{m}}}\zeta_{p}^{Tr(-cx^{d_{2}})}\\
&=2p^{3m}.
\end{split}
\end{equation}
\begin{equation*}
\begin{split}
&\sum_{(a,b,c) \in \mathbb{F}_{p^{m}}^{3}}T^{2}(a,b,c)\\
&=\sum_{(x_{1},x_{2}) \in \mathbb{F}_{p^{m}}^{2}}\sum_{a \in \mathbb{F}_{p^{m}}}\zeta_{p}^{Tr(a(x_{1}^{2}+x_{2}^{2}))}\sum_{b \in \mathbb{F}_{p^{m}}}\zeta_{p}^{Tr(b(x_{1}^{d_{1}}+x_{2}^{d_{1}}))}\sum_{c \in \mathbb{F}_{p^{m}}}\zeta_{p}^{Tr(c(x_{1}^{d_{2}}+x_{2}^{d_{2}}))}\\
&\relphantom{=} {}+\sum_{(x_{1},x_{2}) \in \mathbb{F}_{p^{m}}^{2}}\sum_{a \in \mathbb{F}_{p^{m}}}\zeta_{p}^{Tr(a(x_{1}^{2}-x_{2}^{2}))}\sum_{b \in \mathbb{F}_{p^{m}}}\zeta_{p}^{Tr(b(x_{1}^{d_{1}}-x_{2}^{d_{1}}))}\sum_{c \in \mathbb{F}_{p^{m}}}\zeta_{p}^{Tr(c(x_{1}^{d_{2}}-x_{2}^{d_{2}}))}\\
&\relphantom{=} {}+\sum_{(x_{1},x_{2}) \in \mathbb{F}_{p^{m}}^{2}}\sum_{a \in \mathbb{F}_{p^{m}}}\zeta_{p}^{Tr(a(-x_{1}^{2}+x_{2}^{2}))}\sum_{b \in \mathbb{F}_{p^{m}}}\zeta_{p}^{Tr(b(-x_{1}^{d_{1}}+x_{2}^{d_{1}}))}\sum_{c \in \mathbb{F}_{p^{m}}}\zeta_{p}^{Tr(c(-x_{1}^{d_{2}}+x_{2}^{d_{2}}))}\\
&\relphantom{=} {}+\sum_{(x_{1},x_{2}) \in \mathbb{F}_{p^{m}}^{2}}\sum_{a \in \mathbb{F}_{p^{m}}}\zeta_{p}^{Tr(-a(x_{1}^{2}+x_{2}^{2}))}\sum_{b \in \mathbb{F}_{p^{m}}}\zeta_{p}^{Tr(-b(x_{1}^{d_{1}}+x_{2}^{d_{1}}))}\sum_{c \in \mathbb{F}_{p^{m}}}\zeta_{p}^{Tr(-c(x_{1}^{d_{2}}+x_{2}^{d_{2}}))}\\
&=p^{3m}(\#S_{1}+\#S_{2}+\#S_{3}+\#S_{4}),
\end{split}
\end{equation*}
where
\begin{equation*}
\begin{split}
&S_{1}=\{(x_{1},x_{2}) \in \mathbb{F}_{p^{m}}^{2}: x_{1}^{2}+x_{2}^{2}=0, x_{1}^{d_{1}}+x_{2}^{d_{1}}=0,x_{1}^{d_{2}}+x_{2}^{d_{2}}=0\},\\
&S_{2}=\{(x_{1},x_{2}) \in \mathbb{F}_{p^{m}}^{2}: x_{1}^{2}-x_{2}^{2}=0, x_{1}^{d_{1}}-x_{2}^{d_{1}}=0,x_{1}^{d_{2}}-x_{2}^{d_{2}}=0\},\\
&S_{3}=\{(x_{1},x_{2}) \in \mathbb{F}_{p^{m}}^{2}: -x_{1}^{2}+x_{2}^{2}=0, -x_{1}^{d_{1}}+x_{2}^{d_{1}}=0,-x_{1}^{d_{2}}+x_{2}^{d_{2}}=0\},\\
&S_{4}=\{(x_{1},x_{2}) \in \mathbb{F}_{p^{m}}^{2}: -x_{1}^{2}-x_{2}^{2}=0, -x_{1}^{d_{1}}-x_{2}^{d_{1}}=0,-x_{1}^{d_{2}}-x_{2}^{d_{2}}=0\}.\\
\end{split}
\end{equation*}
It is clear that $S_{1}=S_{4}$ and $S_{2}=S_{3}$. Then by Lemma \ref{Le:3.9} and \ref{Le:3.8}, we have
\begin{equation}\label{Eq:3.14}
\sum_{(a,b,c) \in \mathbb{F}_{p^{m}}^{3}}T^{2}(a,b,c)=4p^{4m}.
\end{equation}
Similarly, by Lemmas \ref{Le:3.10}-\ref{Le:3.14}, we have
\begin{equation}\label{Eq:3.15}
\begin{split}
&\sum_{(a,b,c) \in \mathbb{F}_{p^{m}}^{3}}T^{3}(a,b,c)=8p^{3m}(p^{m+d}+p^{m}-p^{d})\\
&\sum_{(a,b,c) \in \mathbb{F}_{p^{m}}^{3}}T^{4}(a,b,c)=16p^{4m}(p^{m+d}+p^{m}-p^{d}).
\end{split}
\end{equation}
Combining Eqs. (\ref{Eq:3.10})-(\ref{Eq:3.15}), we get
\begin{equation*}
\begin{split}
&n_{1,0}=\frac{(p^{m+d}+p^{(m+3d)/2})(p^{2m}-p^{2m-2d}-p^{2m-3d}+p^{m-2d}+p^{m-3d}-1)}{2(p^{2d}-1)},\\
&n_{1,1}=\frac{(p^{m+d}-p^{(m+3d)/2})(p^{2m}-p^{2m-2d}-p^{2m-3d}+p^{m-2d}+p^{m-3d}-1)}{2(p^{2d}-1)},\\
&n_{2,0}=\frac{(p^{m-3d}+p^{(m-3d)/2})(p^{m}-1)(p^{m-d}-1)}{2(p^{2d}-1)},\\
&n_{2,1}=\frac{(p^{m-3d}-p^{(m-3d)/2})(p^{m}-1)(p^{m-d}-1)}{2(p^{2d}-1)}.
\end{split}
\end{equation*}
Summarizing the discussion above completes the proof of this theorem in the case of $q_{0} \equiv 3 \pmod 4$.

\begin{table}[htbp]
\caption{Weight Distribution of $\mathcal{C}$}
\centering
\begin{tabular}{ll}
 \hline
 Weight& Frequency\\
 \hline
 $0$ &1\\
 $(p^{t}-1)p^{m-t}$ & $(p^{m}-1)(p^{2m}-p^{2m-d}+p^{2m-4d}+p^{m}-p^{m-d}-p^{m-3d}+1)$\\
 $(p^{t}-1)(p^{m-t}-p^{\frac{m+d-2t}{2}})$ &  $\frac{(p^{m+d}+p^{(m+3d)/2})(p^{2m}-p^{2m-2d}-p^{2m-3d}+p^{m-2d}+p^{m-3d}-1)}{2(p^{2d}-1)}$\\
 $(p^{t}-1)(p^{m-t}+p^{\frac{m+d-2t}{2}})$ &  $\frac{(p^{m+d}-p^{(m+3d)/2})(p^{2m}-p^{2m-2d}-p^{2m-3d}+p^{m-2d}+p^{m-3d}-1)}{2(p^{2d}-1)}$\\
 $(p^{t}-1)(p^{m-t}-p^{\frac{m+3d-2t}{2}})$ & $\frac{(p^{m-3d}+p^{(m-3d)/2})(p^{m}-1)(p^{m-d}-1)}{2(p^{2d}-1)}$\\
 $(p^{t}-1)(p^{m-t}+p^{\frac{m+3d-2t}{2}})$ & $\frac{(p^{m-3d}-p^{(m-3d)/2})(p^{m}-1)(p^{m-d}-1)}{2(p^{2d}-1)}$\\
 \hline
\end{tabular}
\end{table}
Recall that $\mathcal{C}$ is the  cyclic code over $\mathbb{F}_{p^{t}}$ with parity check polynomial $h_{0}(x)h_{1}(x)h_{2}(x)$, where $h_{0}(x)$, $h_{1}(x)$, $h_{2}(x)$ are the minimal polynomial of $\pi^{-2}$, $\pi^{-(p^{k}+1)}$ and $\pi^{-(p^{2k}+1)}$ over $\mathbb{F}_{p^{t}}$, respectively.
\begin{theorem}\label{Th:3.11}
Let $m$ and $k$ be any two  positive integers such that $s=m/d \geq 5$ is odd, where $d=$ gcd$(m,k)$.  Let $t$ be a divisor of $d$ such that $d/t$ is odd, then
$\mathcal{C}$ is a  cyclic code over $\mathbb{F}_{p^{t}}$ with parameters $[p^{m}-1, 3m_{0}, (p^{t}-1)(p^{m-t}-p^{\frac{m+3d-2t}{2}})]$, where $m_{0}=m/t$. Moreover, the weight distribution of $\mathcal{C}$ is given in Table 2.
\end{theorem}
\begin{proof}
According to the discussion in the beginning of this section,
the length and dimension of $\mathcal{C}$ are clear. Furthermore,
 the minimum Hamming distance and weight distribution of $\mathcal{C}$ follows from Eq. (\ref{Eq:3.1}) and Theorem \ref{Th:3.7}.
\end{proof}
Finally, we give an example to verify the results in Table 2. But the experiment for the case $d\neq 1$  is beyond our computation ability.
\begin{example}
Let $p=3$, $m=7$ and $k=1$. Then the code $\mathcal{C}$ is a $[2186,21,1296]$ cyclic code over $\mathbb{F}_{3}$ with weight enumerator
\[
\begin{split}
&1+8951670z^{1296}+1732767876z^{1404}+7102473578z^{1458}+1608998742z^{1512}\\
&+7161336z^{1620}.
\end{split}
\]
which is completely in agreement with the results presented in Table 2.
\end{example}
\begin{center}
    APPENDIX
\end{center}
\noindent $\mathbf{Proof}$ $\mathbf{of}$ $\mathbf{Lemma}$ \ref{Le:3.12}.

\noindent
For any $(\overline{a},\overline{b},\overline{c}) \in \mathbb{F}_{p^{m}}^{3}$, let $N_{1(\overline{a},\overline{b},\overline{c})}$ and $N_{2(\overline{a},\overline{b},\overline{c})}$ denote the number of solutions  of the following two systems of equations
\begin{equation}\label{4.16}
\begin{cases}
x_{1}^{2}+x_{2}^{2}=\overline{a}\\
x_{1}^{d_{1}}+x_{2}^{d_{1}}=\overline{b}\\
x_{1}^{d_{2}}+x_{2}^{d_{2}}=\overline{c}\\
\end{cases}
\end{equation}
\begin{equation}\label{4.17}
\begin{cases}
x_{3}^{2}+x_{4}^{2}=-\overline{a}\\
x_{3}^{d_{1}}+x_{4}^{d_{1}}=-\overline{b}\\
x_{3}^{d_{2}}+x_{4}^{d_{2}}=-\overline{c}.\\
\end{cases}
\end{equation}
Then we have
 \[N_{4}=\sum_{(\overline{a},\overline{b},\overline{c}) \in \mathbb{F}_{p^{m}}^{3}}N_{1(\overline{a},\overline{b},\overline{c})}N_{2(\overline{a},\overline{b},\overline{c})}.\]
\noindent Case 1, when $\overline{a}=0$. In this case, (\ref{4.16}) and (\ref{4.17}) have solutions if and only if $\overline{b}=\overline{c}=0$  since $-1$ is a nonsquare. Moreover, $N_{1(0,0,0)}=N_{2(0,0,0)}=1$.

\noindent Case 2, when $\overline{a}\neq 0$. In this case,  if $\overline{b}=0$ or $\overline{c}=0$, neither (\ref{4.16}) nor (\ref{4.17}) has solutions. So in the following, we consider the problem only when $\overline{b}\neq 0$ and $\overline{c}\neq 0$.
\begin{itemize}
  \item $\overline{a}$ is a nonzero square in $\mathbb{F}_{p^{m}}$, $\overline{b}\neq 0$ and $\overline{c}\neq 0$.
 In this case, for any fixed $\overline{a}$, (\ref{4.16}) has the same number of  solutions as
\begin{equation}\label{4.18}
\begin{cases}
x_{1}^{2}+x_{2}^{2}=1\\
x_{1}^{d_{1}}+x_{2}^{d_{1}}=b\\
x_{1}^{d_{2}}+x_{2}^{d_{2}}=c\\
\end{cases}
\end{equation}
and  (\ref{4.17}) has the same number of  solutions as
\begin{equation}\label{4.19}
\begin{cases}
x_{3}^{2}+x_{4}^{2}=-1\\
x_{3}^{d_{1}}+x_{4}^{d_{1}}=-b\\
x_{3}^{d_{2}}+x_{4}^{d_{2}}=-c,\\
\end{cases}
\end{equation}
where $b=\overline{b}/\overline{a}^{\frac{d_{1}}{2}}$ and $c=\overline{c}/\overline{a}^{\frac{d_{2}}{2}}$. Clearly, $(b,c)$ runs through $\mathbb{F}_{p^{m}}^{*2}$ as $(\overline{b},\overline{c})$ does. According to the proofs of  Lemma \ref{4.1} and \ref{4.2},  we can get $N_{1(1,b,c)}=N_{2(1,b,c)}$ for any fixed $(b,c) \in \mathbb{F}_{p^{m}}^{*2}$ . If (\ref{4.18}) has solutions, then $N_{1(1,b,c)}=p^d+1$ or $2(p^d+1)$. Furthermore, only in the case of $(b,c)=(1,1)$, $N_{1(1,b,c)}=p^d+1$, and there are $\frac{p^{m}-p^d}{2(p^d+1)}$ pairs of $(b,c)$ such that $N_{1(1,b,c)}=2(p^d+1)$. Therefore, for any fixed nonzero square $\overline{a}$, we have
\begin{equation*}
\begin{split}
&\sum_{(\overline{b},\overline{c}) \in \mathbb{F}_{p^{m}}^{*2}}N_{1(\overline{a},\overline{b},\overline{c})}N_{2(\overline{a},\overline{b},\overline{c})}\\
&=(p^d+1)^{2}+(2(p^d+1))^{2}\frac{p^{m}-p^d}{2(p^d+1)}\\
&=(p^d+1)(2p^{m}-p^d+1).
\end{split}
\end{equation*}
  \item $\overline{a}$ is a non-square in $\mathbb{F}_{p^{m}}$.
 In this case, for any fixed $\overline{a}$,  (\ref{4.16}) has the same number of  solutions as
\begin{equation*}
\begin{cases}
x_{1}^{2}+x_{2}^{2}=-1\\
x_{1}^{d_{1}}+x_{2}^{d_{1}}=-b\\
x_{1}^{d_{2}}+x_{2}^{d_{2}}=-c\\
\end{cases}
\end{equation*}
and  equation system (\ref{4.17}) has the same number of  solutions as
\begin{equation*}
\begin{cases}
x_{3}^{2}+x_{4}^{2}=1\\
x_{3}^{d_{1}}+x_{4}^{d_{1}}=b\\
x_{3}^{d_{2}}+x_{4}^{d_{2}}=c.\\
\end{cases}
\end{equation*}
It can be easily seen that this case is equivalent to the case when $\overline{a}$ is a nonzero square. So for any fixed nonsquare $\overline{a}$, we also have
\begin{equation*}
\begin{split}
&\sum_{(\overline{b},\overline{c}) \in \mathbb{F}_{p^{m}}^{*2}}N_{1(\overline{a},\overline{b},\overline{c})}N_{2(\overline{a},\overline{b},\overline{c})}\\
&=(p^d+1)(2p^{m}-p^d+1).
\end{split}
\end{equation*}
\end{itemize}
Summarizing  the two cases above, we have $N_{4}=1+(p^{m}-1)(p^d+1)(2p^{m}-p^d+1)$.

\noindent$\mathbf{Proof}$ $\mathbf{of}$ $\mathbf{Lemma}$ \ref{Le:3.13}.

\noindent
For any $(\overline{a},\overline{b},\overline{c}) \in \mathbb{F}_{p^{m}}^{3}$, let $N_{1(\overline{a},\overline{b},\overline{c})}$ and $N_{3(\overline{a},\overline{b},\overline{c})}$ denote the number of solutions  of the following two system of equations
\begin{equation}\label{4.28}
\begin{cases}
x_{1}^{2}+x_{2}^{2}=\overline{a}\\
x_{1}^{d_{1}}+x_{2}^{d_{1}}=\overline{b}\\
x_{1}^{d_{2}}+x_{2}^{d_{2}}=\overline{c}\\
\end{cases}
\end{equation}
\begin{equation}\label{4.29}
\begin{cases}
x_{3}^{2}-x_{4}^{2}=-\overline{a}\\
x_{3}^{d_{1}}-x_{4}^{d_{1}}=-\overline{b}\\
x_{3}^{d_{2}}-x_{4}^{d_{2}}=-\overline{c}.\\
\end{cases}
\end{equation}
It is then obvious that \[\overline{N_{4}}=\sum_{(\overline{a},\overline{b},\overline{c}) \in \mathbb{F}_{p^{m}}^{3}}N_{1(\overline{a},\overline{b},\overline{c})}N_{3(\overline{a},\overline{b},\overline{c})}.\]

\noindent Case 1, when $\overline{a}=0$.
In this case, (\ref{4.28}) have solutions if and only if $\overline{b}=\overline{c}=0$  since $-1$ is a non-square. Moreover, $N_{1(0,0,0)}=1$ and $N_{3(0,0,0)}=2p^{m}-1$.

\noindent Case 2, when $\overline{a}\neq 0$. In this case,  if $\overline{b}=0$ or $\overline{c}=0$,  (\ref{4.28}) has no solution. So in the following, we consider this problem only when $\overline{b}\neq 0$ and $\overline{c}\neq 0$.
\begin{itemize}
  \item when $\overline{a}$ is a nonzero square, $\overline{b}\neq 0$ and $\overline{c}\neq 0$. In this case, for any fixed $\overline{a}\neq 0$, equation system (\ref{4.28}) has the same number of  solutions as
\begin{equation}\label{4.30}
\begin{cases}
x_{1}^{2}+x_{2}^{2}=1\\
x_{1}^{d_{1}}+x_{2}^{d_{1}}=b\\
x_{1}^{d_{2}}+x_{2}^{d_{2}}=c\\
\end{cases}
\end{equation}
and  equation system (\ref{4.29}) has the same number of  solutions as
\begin{equation}\label{4.31}
\begin{cases}
x_{3}^{2}-x_{4}^{2}=-1\\
x_{3}^{d_{1}}-x_{4}^{d_{1}}=-b\\
x_{3}^{d_{2}}-x_{4}^{d_{2}}=-c,\\
\end{cases}
\end{equation}
where $b=\overline{b}/\overline{a}^{\frac{d_{1}}{2}}$ and $c=\overline{c}/\overline{a}^{\frac{d_{2}}{2}}$.  Then $(b,c)$ runs through $\mathbb{F}_{p^{m}}^{*2}$ as $(\overline{b},\overline{c})$ runs through $\mathbb{F}_{p^{m}}^{*2}$. According to the proofs of Lemma \ref{4.1} and \ref{4.3}, in order to guarantee (\ref{4.30}) and (\ref{4.31}) have solutions simultaneously for any fixed $(b,c)$, we need to prove that the element $c_{1}$ determined by $b$ in (\ref{4.30}) and  the element $c_{2}$ determined by $b$ in (\ref{4.31}) are the same number. By easy calculation, we have  $c_{1}=c_{2}$  if and only if $b=1$. And then $c_{1}=c_{2}=1$. Furthermore, for any fixed nonzero square $\overline{a}$,
\begin{equation*}
\begin{split}
&\sum_{(\overline{b},\overline{c}) \in \mathbb{F}_{p^{m}}^{*2}}N_{1(\overline{a},\overline{b},\overline{c})}N_{3(\overline{a},\overline{b},\overline{c})}\\
&=(p^d+1)(p^d-1).
\end{split}
\end{equation*}
  \item when $\overline{a}$ is a non-square, $\overline{b}\neq 0$ and $\overline{c}\neq 0$. In this case, for any fixed $\overline{a}\neq 0$, equation system (\ref{4.28}) has the same number of  solutions as
\begin{equation*}
\begin{cases}
x_{1}^{2}+x_{2}^{2}=-1\\
x_{1}^{d_{1}}+x_{2}^{d_{1}}=-b\\
x_{1}^{d_{2}}+x_{2}^{d_{2}}=-c\\
\end{cases}
\end{equation*}
and  equation system (\ref{4.29}) has the same number of  solutions as
\begin{equation*}
\begin{cases}
x_{3}^{2}-x_{4}^{2}=1\\
x_{3}^{d_{1}}-x_{4}^{d_{1}}=b\\
x_{3}^{d_{2}}-x_{4}^{d_{2}}=c.\\
\end{cases}
\end{equation*}
It can be easily seen that this case is equivalent to the case when $\overline{a}$ is a nonzero square. So  for any fixed non-square $\overline{a}$, we also have
\begin{equation*}
\begin{split}
&\sum_{(\overline{b},\overline{c}) \in \mathbb{F}_{p^{m}}^{*2}}N_{1(\overline{a},\overline{b},\overline{c})}N_{3(\overline{a},\overline{b},\overline{c})}\\
&=(p^d+1)(p^d-1).
\end{split}
\end{equation*}
\end{itemize}
\[\overline{N_{4}}=\sum_{(\overline{a},\overline{b},\overline{c}) \in \mathbb{F}_{p^{m}}^{3}}N_{1(\overline{a},\overline{b},\overline{c})}N_{3(\overline{a},\overline{b},\overline{c})}=(2p^m-1)+(p^m-1)(p^d+1)(p^d-1).\]

 \noindent$\mathbf{Proof}$ $\mathbf{of}$ $\mathbf{Lemma}$ \ref{Le:3.14}.

\noindent
With the notation as above,
 \[\widetilde{N_{4}}=\sum_{(\overline{a},\overline{b},\overline{c}) \in \mathbb{F}_{p^{m}}^{3}}N_{1(\overline{a},\overline{b},\overline{c})}^{2}=1+(p^{m}-1)(p^{d}+1)(2p^{m}-p^{d}+1).\]

\begin{lemma}\label{4.1}
Let $N_{1(b,c)}$ denote the number of solutions $(x_{1}, x_{2}) \in \mathbb{F}_{p^{m}}^{2}$ of (\ref{4.18}), where $(b,c) \in \mathbb{F}_{p^{m}}^{*2}$. Then we have the following conclusions.

\noindent (1). $N_{1(1,1)}=p^d+1$.

\noindent(2). When $(b,c)$ runs through $\mathbb{F}_{p^{m}}^{*2}\setminus\{(1,1)\}$,
\begin{equation*}
N_{1(b,c)}=
\begin{cases}
2(p^d+1),& for~ \frac{p^{m}-p^{d}}{2(p^d+1)}~ times,\\
0, & for~ the ~rest.
\end{cases}
\end{equation*}
\end{lemma}
\begin{proof}
We first compute the number $N_{1(b)}$ of solutions $(x_{1}, x_{2}) \in \mathbb{F}_{p^{m}}^{2}$ of the following system of equations
\begin{equation}\label{20}
\begin{cases}
x_{1}^{2}+x_{2}^{2}=1\\
x_{1}^{d_{1}}+x_{2}^{d_{1}}=b.\\
\end{cases}
\end{equation}
When $q_{0}\equiv 3 \pmod 4$, $-1$ is a non-square in $ \mathbb{F}_{p^{m}}$. Then we can
choose $t \in \mathbb{F}_{p^{2m}}$ such that $t^{2}=-1$. From the first equation of (\ref{20}), by setting $\theta =x_{1}-tx_{2} \in \mathbb{F}_{p^{2m}}^{*}$, we have
\begin{equation}\label{21}
x_{1}=\frac{\theta + \theta^{-1}}{2},
x_{2}=\frac{t(\theta - \theta^{-1})}{2}.
\end{equation}
Since $x_{1} \in  \mathbb{F}_{p^{m}}$,  the following holds:
\begin{equation*}
\frac{\theta+\theta^{-1}}{2}=(\frac{\theta+\theta^{-1}}{2})^{p^{m}}=\frac{\theta^{p^{m}}+\theta^{-p^{m}}}{2},
\end{equation*}
which implies $\theta^{p^{m}+1}=1$ or $\theta^{p^{m}-1}=1$. If $\theta^{p^{m}+1}\neq 1$, then $\theta^{p^{m}-1}=1$. In this case,  $\theta \in \mathbb{F}_{p^{m}}^{*}$.  Since $x_{2}=\frac{t(\theta - \theta^{-1})}{2} \in \mathbb{F}_{p^{m}}^{*}$, we have $t \in \mathbb{F}_{p^{m}}^{*}$, which is a contradiction. Hence $\theta^{p^{m}+1}=1$.
\begin{itemize}
  \item When k is even, we have  $p^k+1 \equiv 2 \pmod 4$, then $t^{(p^k+1)}=-1$. Substituting (\ref{21}) into the second equation of (\ref{20}), we obtain
\begin{equation}\label{22}
\theta^{p^{k}-1}+\theta^{1-p^{k}}=2b.
\end{equation}
Denote $\theta^{p^{k}-1}$ by $w$, Eq. (\ref{22}) is equivalent to
\begin{equation}\label{23}
w^{2}-2bw+1=0.
\end{equation}
If Eq.(\ref{23}) has no solution, i.e., $b^{2}-1$ is a non-square of $\mathbb{F}_{p^{2m}}^{*}$, then $N_{b}=0$. Otherwise, let $w_{1}$ and $w_{2}=w_{1}^{-1}$ be two solutions of (\ref{23}). According to the discussion above, we have
      \begin{equation}\label{24}
      \theta^{p^{k}-1}=w_{1}, \theta^{p^{m}+1}=1,
      \end{equation} or
      \begin{equation}\label{25}
      \theta^{p^{k}-1}=w_{1}^{-1}, \theta^{p^{m}+1}=1.
      \end{equation}
      If $\theta_{1}$ and $\theta_{2}$ are two solutions of (\ref{24}), then $(\theta_{1}/\theta_{2})^{p^{k}-1}=1=(\theta_{1}/\theta_{2})^{p^{m}+1}$. Since gcd$(p^{k}-1, p^{m}+1)=p^d+1$, then $(\theta_{1}/\theta_{2})^{p^d+1}=1$. So if (\ref{24}) has solutions, then it has exactly $p^d+1$ solutions.
      \begin{itemize}
        \item If $w_{1}=w_{1}^{-1}$, then (\ref{25}) is the same as (\ref{24}). In this case we have $w_{1}=\pm 1$ and then from Eq.(\ref{23}), $b=\pm 1$. But when $b=-1$, $\theta^{p^{k}-1}=w_{1}=-1$. By $\theta^{p^{m}+1}=1$ and gcd$(2(p^{k}-1), p^{m}+1)=p^d+1$, we have $\theta^{p^d+1}=1$. And then $\theta^{p^{k}-1}=1$, which is a contradiction. So  we only consider $b=1$, which implies $w_{1}=1$. Then ($\ref{24}$) and ($\ref{25}$) both have $p^d+1$ solutions. As a result, we have $p^d+1$ solutions of (\ref{20}).
        \item If $w_{1}\neq w_{1}^{-1}$, then (\ref{25}) has the same number of solutions as $(\ref{24})$. Moreover, their solutions are distinct since $w_{1} \neq \pm1$. Therefore, ($\ref{24}$) and ($\ref{25}$) both have $p^d+1$ solutions or no solutions in $\mathbb{F}_{p^{2m}}$. By (\ref{21}), $(x_{1}, x_{2})$ is uniquely determined by $\theta$. Then (\ref{20}) has $2(p^d+1)$ solutions or no solutions in $\mathbb{F}_{p^{m}}^{2}$.
      \end{itemize}
        Until now, we have  $N_{1(1)}=p^d+1$ and $N_{1(b)}=0$ or $2(p^d+1)$ for $b\neq 1$. And as in Lemma 5.4 in \cite{20},
we define
\begin{equation*}
T=\#\{b \in \mathbb{F}_{p^{m}}: N_{1(b)}=2(p^d+1)\}.
\end{equation*}
Then we have
  \begin{equation*}
T=\frac{p^{m}-p^d}{2(p^d+1)}.
\end{equation*}

Substituting (\ref{21}) into the third equation of (\ref{4.18}), we obtain
\begin{equation}\label{27}
\theta^{p^{2k}-1}+\theta^{1-p^{2k}}=2c,
\end{equation}
which implies $c=\frac{1}{2}\{{(b+\sqrt{b^2-1})^{p^k+1}+(b-\sqrt{b^2-1})^{p^k+1}}\}$. Hence if (\ref{4.18}) has solutions, then $N_{1(b,c)}=N_{1(b)}$ and $c$ is uniquely determined by $b$.
  \item When k is odd, we have $p^k+1 \equiv 0 \pmod 4$, then $t^{(p^k+1)}=1$. Similarly, if (\ref{4.18}) has solutions, then we have $N_{1(b,c)}=N_{1(b)}$ and $c$ is uniquely determined by $b$.
\end{itemize}
Summarizing all the cases above completes the proof.
\end{proof}
\begin{lemma}\label{4.2}
Let $N_{2(b,c)}$ denote the number of solutions $(x_{1}, x_{2}) \in \mathbb{F}_{p^{m}}^{2}$ of (\ref{4.19}), where $(b,c) \in \mathbb{F}_{p^{m}}^{*2}$. Then we have the following conclusions.

\noindent (1). $N_{2(1,1)}=p^d+1$.

\noindent(2). When $(b,c)$ runs through $\mathbb{F}_{p^{m}}^{*2}\setminus\{(1,1)\}$,
\begin{equation*}
N_{2(b,c)}=
\begin{cases}
2(p^d+1),& for~ \frac{p^{m}-p^d}{2(p^d+1)}~ times,\\
0, & for~ the ~rest.
\end{cases}
\end{equation*}
\end{lemma}
\begin{proof}
The proof is similar to the proof of the lemma above.
\end{proof}
\begin{lemma}\label{4.3}

Let $N_{3(b,c)}$ denote the number of solutions $(x_{1}, x_{2}) \in \mathbb{F}_{p^{m}}^{2}$ of (\ref{4.31}), where $(b,c) \in \mathbb{F}_{p^{m}}^{*2}$. Then we have the following conclusions.

\noindent (1). $N_{3(1,1)}=p^d-1$.

\noindent(2). When $(b,c)$ runs through $\mathbb{F}_{p^{m}}^{*2}\setminus\{(1,1)\}$,
\begin{equation*}
N_{3(b,c)}=
\begin{cases}
2(p^d-1),& for~ \frac{p^{m}-p^{d}}{2(p^d-1)}~ times,\\
0, & for~ the ~rest.
\end{cases}
\end{equation*}
\end{lemma}
\begin{proof}
The proof is similar to the proof of the Lemma \ref{4.1}. 
\end{proof}

\noindent $\mathbf{Acknowledgement}$:  The authors are very grateful to the Editor in Chief, the Coordinating Editor and the anonymous reviewers, for their helpful comments that improved the quality of this paper.


\begin{thebibliography}{40}
\bibitem{1}
{L.D. Baumert, R.J. McEliece}, ¡°Weights of irreducible cyclic codes¡±,
{\it Inf. Contr.}, {\bf 20, no. 2}~(1972), 158-175.
\bibitem{2}
{L.D. Baumert, J. Mykkeltveit}, ¡°Weight distribution of some irreducible cyclic codes¡±,
{\it DSN Progr. Rep.},  {\bf 16}~(1973),  128-131.

\bibitem{3}
{A.R. Calderbank, J.M. Goethals}, ¡°Three-weight codes and association schemes¡±,
{\it Philips J. Res.},  {\bf 39}~(1984), 143-152.
\bibitem{4}
{C. Carlet, C. Ding, J. Yuan}, ¡°Linear codes from highly nonlinear functions and their secret sharing schemes¡±,
{\it IEEE Trans. Inf. Theory},  {\bf 51, no.6}~(2005), 2089-2102.
\bibitem{6}
{C. Ding, T. Helleseth}, ¡°Optimal ternary cyclic codes from monomials¡±,
{\it IEEE Trans. Inf. Theory},  {\bf 313, no. 4}~(2013), 5898-5904.
\bibitem{7}
{C. Ding, Y. Liu, C. Ma, L. Zeng}, ¡°The weight distributions of the duals of cyclic codes with two zeros¡±,
{\it IEEE Trans. Inf. Theory},  {\bf 57, no. 12}~(2011), 8000-8006.
\bibitem{9}
{K. Feng, J. Luo}, ¡°Weight distribution of some reducible cyclic codes¡±,
{\it Finite Fields Appl.},  {\bf 14, no. 2}~(2008), 390-409.
\bibitem{18}
{T. Feng, K. Leung, Q. Xiang}, ¡°Binary cyclic codes with two primitive nonzeros¡±,
{\it Sci. China Math.},  {\bf 56, no. 7}~(2012), 1403-1412.
\bibitem{28}
{H.D.L. Hollmann, Q. Xiang}, ¡°On binary codes with few weights¡±,
{\it Finite Fields Appl.},  {\bf 11, no. 1}~(2005), 89-110.
\bibitem{29}
{C. Li, X. Zeng, L. Hu}, ¡°A class of binary cyclic codes with five weights¡±,
{\it Sci. China Math.},  {\bf 53, no. 12}~(2010), 3279-3286.
\bibitem{11}
{R. Lidl, H. Niederreiter}, ¡°Finite fieds¡±,
{\it Addison-Wdsley Publishing Inc.},  (1983).
\bibitem{16}
{Y. Liu, H. Yan, C. Liu}, ¡°A class of six-weight cyclic codes and their weight distribution¡±,
{\it Des. Codes Crypogr.},  (2014), doi: 10.1007/s10623-014-9984-y.
\bibitem{12}
{J. Luo, K. Feng}, ¡°On the weight distribution of two classes of  cyclic codes¡±,
{\it IEEE Trans. Inf. Theory},  {\bf 54, no. 12}~(2008), 5332-5344.
\bibitem{13}
{C. Ma, L. Zeng, Y. Liu, D. Feng, C. Ding}, ¡°The weight enumerator of a class of cyclic codes¡±,
{\it IEEE Trans. Inf. Theory},  {\bf 57, no. 1}~(2011), 397-402.

\bibitem{14}
{J. Yuan, C. Carlet, C. Ding}, ¡°The weight distribution of a class of linear codes from perfect nonlinear functions¡±,
{\it IEEE Trans. Inf. Theory},  {\bf 52, no. 2}~(2006), 712-717.
\bibitem{15}
{B. Wang, C. Tang, Y. Qi, Y. Yang, M. Xu}, ¡°The weight distributions of cyclic codes¡±,
{\it IEEE Trans. Inf. Theory},  {\bf 40, no. 6}~(1994), 2068-2071.
\bibitem{30}
{J. Wolfmann}, ¡°Weight distributions of some binary primitive cyclic codes and elliptic curves¡±,
{\it IEEE Trans. Inf. Theory},  {\bf 58, no. 12}~(2012), 7253-7259.
\bibitem{21}
{M. Xiong}, ¡°The weight distributions of a class of cyclic codes¡±,
{\it Finite Fields Appl.},  {\bf 18, no. 5}~(2012), 933-945.
\bibitem{22}
{D. Zheng, X. Wang, X. Zeng, L. Hu}, ¡°The weight distribution of a famimily of $p$-ary cyclic codes¡±,
{\it Des. Codes Crypogr.},  (2013), doi: 10.1007/s10623-013-9908-2.

\bibitem{20}
{Z. Zhou, C. Ding, J. Luo, A. Zhang}, ¡°A family of five-weight cyclic codes and their weight enumerators¡±,
{\it IEEE Trans. Inf. Theory},  {\bf 59, no. 10}~(2013), 6674-6682.











\end{thebibliography}
\end{document}